\documentclass[journal]{IEEEtran}

\usepackage{booktabs}
\usepackage{amsthm}

\usepackage{ifpdf}

\newtheorem{thm}{Theorem}

\theoremstyle{remark}

\usepackage{cite}

\usepackage[dvips]{graphicx}
\usepackage{epstopdf}

\usepackage{lipsum}

\usepackage[cmex10]{amsmath}
\usepackage{algorithmic}

\usepackage{array}

\ifCLASSOPTIONcompsoc
\usepackage[caption=false,font=normalsize,labelfont=sf,textfont=sf]{subfig}
\else
\usepackage[caption=false,font=footnotesize]{subfig}
\fi

\begin{document}
%
\title{Optimal Distribution Design for Irregular Repetition Slotted ALOHA with Multi-Packet Reception
}
%
%
%

\author{Zhengchuan Chen, \IEEEmembership{Member, IEEE,} Yifan Feng, \IEEEmembership{Student Member, IEEE,} \\
Chundie Feng, \IEEEmembership{Student Member, IEEE,} Liang Liang, \IEEEmembership{Member, IEEE,}\\
Yunjian Jia, \IEEEmembership{Member, IEEE,} and Tony Q. S. Quek, \IEEEmembership{Fellow, IEEE}
\thanks{Z. Chen, Y. Feng, C. Feng, L. Liang and Y. Jia are with the College of Microelectronics and Communication Engineering, Chongqing University, Chongqing 400044, China (email: \{czc, fengyf, fcd, liangliang, yunjian\}@cqu.edu.cn).}
\thanks{T. Q. S. Quek is with the Information Systems Technology and Design Pillar, SUTD, Singapore 487372 (email: tonyquek@sutd.edu.sg).}}

\maketitle

\begin{abstract}
Associated with multi-packet reception at the access point, irregular repetition slotted ALOHA (IRSA) holds a great potential in improving the access capacity of massive machine type communication systems. Considering the time-frequency resource efficiency, $K = 2$ (multi-packet reception capability) may be the most suitable scheme for scenarios that allow smaller resource efficiency in exchange for greater throughput. In this paper, we analytically derive an optimal transmission probability distribution for IRSA with $K = 2$, which achieves a significant higher load threshold than the existing benchmark distributions. In addition, the energy efficiency optimization in terms of the maximum repetition rate is also presented.
\end{abstract}


%
\IEEEpeerreviewmaketitle

\section{Introduction}
\allowdisplaybreaks
\IEEEPARstart{D}{ue} to the unique advantage of random access schemes in terms of signaling consumption, ALOHA-type protocols are considered to be a class of promising access technologies for massive machine type communications (mMTC) systems. In recent years, a number of fruitful works have gradually increases the throughput to $1$ packet$/$slot in the evolution direction of slotted ALOHA (SA) \cite{SA,DSA,CRDSA,CRDSA++,IRSA,CSA}. Among them, contention resolution diversity slotted ALOHA (CRDSA) \cite{CRDSA} was the first to introduce successive interference cancelation (SIC) to make use of collisions, instead of directly discarding collision packets like the previous schemes. While allowing each active user to transmit two packet replicas in random slots of the MAC frame, the receiver (access point) in CRDSA scheme iteratively decodes the packets in singleton slots and removes the corresponding replicas from other slots. The resulting time diversity gain helps the throughput reach $0.55$ packet$/$slot (while the value is $1/e$ for SA).

In the subsequent irregular repetition slotted ALOHA (IRSA) scheme \cite{IRSA}, the number of replicas transmitted by each user was designed to depend on a pre-designed probability distribution instead of being fixed to two. In particular, the asymptotic throughput can approach the theoretical upper bound ($1$ packet$/$slot) for ALOHA-type protocols by using an excellent distribution as given by \cite{K=1}. Based on that, coded slotted ALOHA (CSA) \cite{CSA} further increased the rate (valid symbols$/$total symbols) limit from $1/2$ of IRSA to $1$ by segmenting and encoding the packets to be transmitted.

A basic assumption of the above schemes is that the receiver can only decode one packet in a slot at most. With support of physical layer technologies such as power capture and multi-antenna transmission, further throughput improvements can be achieved by considering multi-packet reception (MPR) channels, which has actually been introduced into IRSA and CSA \cite{Mulit3,Mulit4,Mulit5,Mulit1,Mulit2}. The receiver in a $K$-MPR model is assumed to be able to successfully decode all packets in a slot containing no more than $K$ packets. Assuming that the throughput gain comes from a $K$-fold time-frequency resource consumption, \cite{Mulit1} and \cite{Mulit2} both studied the trend of $G^*/K$ (called the normalized load threshold) with respect to $K$. Although contrary to the corresponding conclusion for CSA in \cite{Mulit2}, it is shown in \cite{Mulit1} that $G^*/K$ decreases with $K$ in the asymptotic analysis for IRSA. Based on IRSA scheme, $K = 2$ is obviously the optimal choice for the cases that allow a small penalty in resource efficiency to obtain a significant throughput gain.

Unlike most existing transmission probability distributions that are obtained by search algorithms, \cite{K=1} uses a traceable method to provide an optimal distribution for the conventional IRSA scheme (i.e., $K = 1$). In this paper, we analytically derive an optimal transmission distribution for IRSA with $K = 2$ from the perspective of physical meaning. For $K \ge 3$, some potentially enlightening discussions for finding the corresponding distributions are provided. It is worth mentioning that the derived optimal distribution is actually a discrete function about the maximum repetition rate (i.e., the maximum number of the packet replicas). In order to characterize the impact of the maximum repetition rate more comprehensively, we also present the energy efficiency optimization.

The rest of this paper is organized as follows. In Section \uppercase\expandafter{\romannumeral2}, we present the system model. The derivation of the optimal distribution for IRSA with $K=2$ and the discussions for the case of $K \ge 3$ are given in Section \uppercase\expandafter{\romannumeral3}. Section \uppercase\expandafter{\romannumeral4} optimizes the energy efficiency in terms of the maximum repetition rate for the derived distribution. Based on these, numerical results are provided in Section \uppercase\expandafter{\romannumeral5}. Section \uppercase\expandafter{\romannumeral6} concludes the paper.

\section{System Model and Preliminaries}
\allowdisplaybreaks
\begin{figure*}
\centering
\includegraphics[width=0.8\textwidth]{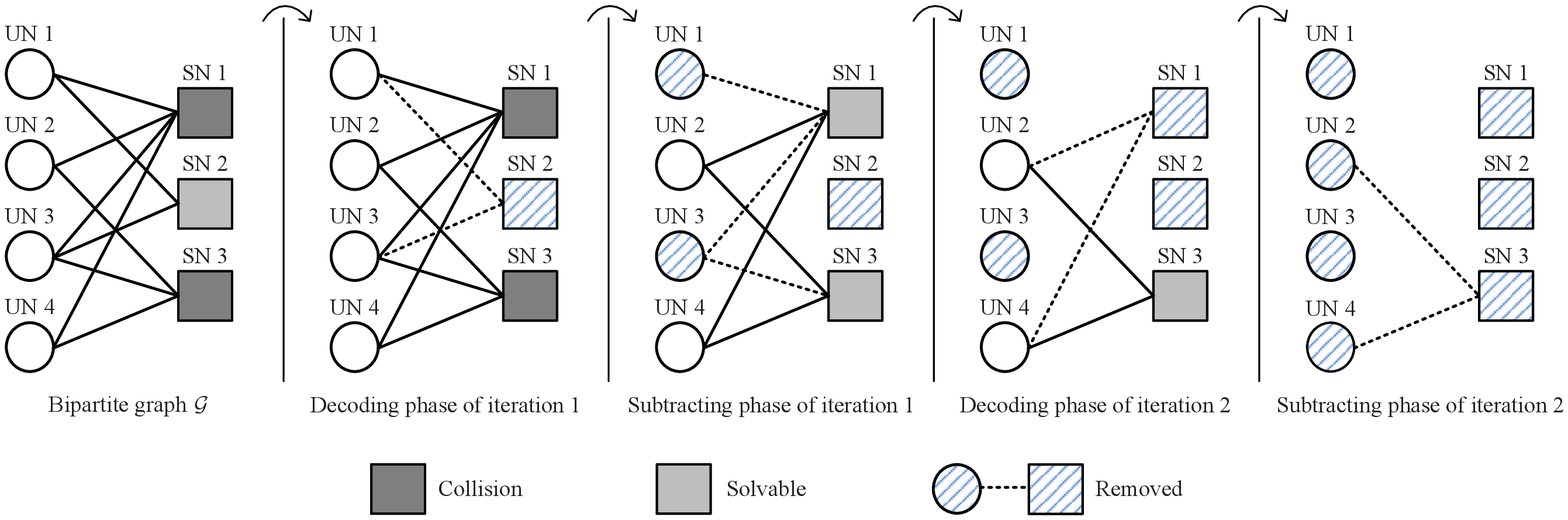}
\caption{An example of SIC with $K=2$. In the decoding phase of iteration 1, SN 2 can be removed. Then all edges emanating from UN 1 and UN 3 are removed in the subtracting phase of iteration 1. Repeat the above steps, and finally all users are successfully decoded after iteration 2.}
\end{figure*}

Following the IRSA model in \cite{IRSA} and \cite{Mulit1}, we consider that there are $M$ active users who want to transmit their packets to the same receiver. Assume that one MAC frame is composed of $N$ time slots and the transmission time of each packet is equal to the duration of a slot. In order to use SIC against collisions, each user transmits multiple repeated packets (i.e., replicas) according to the same transmission probability distribution per frame. More precisely, each user randomly selects a number of slots from a frame where the number is generated according to the pre-designated distribution to transmit this number of replicas. Note that in a slot, there is at most one packet belonging to the same user. At the receiver, we assume that when the number of collision packets in a slot does not exceed $K$, the packets can be decoded. The traffic load is defined as $G = M/N$. If all users can be successfully decoded, we consider the throughput $T = G$.

To facilitate analysis, a bipartite graph $\mathcal{G} = (V, Y, Z)$ is introduced to represent IRSA scheme in \cite{IRSA}, where the set $V$ denotes $M$ user nodes (UNs), $Y$ denotes $N$ slot nodes (SNs) and $Z$ denotes edges. An edge connecting UN 1 and SN 3 can be abstractly understood as a replica of user 1 transmitted in slot 3. During SIC, whenever a packet is decoded, we remove the corresponding edge and the UN connected to the edge (including all edges emanating from this UN) in graph $\mathcal{G}$. Fig. 1 shows an example of SIC, where the numbers of UNs and SNs are $4$ and $3$, respectively.

The number of edges connected to a node is called node degree. Furthermore, the UN degree distribution and SN degree distribution are shown respectively by $\Lambda (x) = \sum\nolimits_r {{\Lambda _r}{x^r}}$ and $\Psi (x) = \sum\nolimits_l {{\Lambda _l}{x^l}}$, where $\Lambda _r$ denotes the probability that the degree of each user node is $r$ and $\Psi _l$ denotes the probability that the degree of each slot node is $l$. It is natural to know that UN degree distribution $\Lambda (x)$ is equivalent to the transmission distribution of users, which totally depends on the system designer's setting and further determines SN degree distribution $\Psi (x)$. In addition, the degree distributions can be defined from an edge perspective as $\lambda (x)  = \sum\nolimits_r {{\lambda _r}{x^{r-1}}}$ and $\rho (x) = \sum\nolimits_l {{\rho _l}{x^{l-1}}}$, where $\lambda _r$ and $\rho _l$ are the probabilities that an edge is connected to a degree-$r$ UN and a degree-$l$ SN, respectively. According to \cite{IRSA}, for $M,N \to \infty$, we have
\begin{align}
\label{edge_distribution}
	\lambda \left( x \right) = \frac{{\Lambda '\left( x \right)}}{{\Lambda '\left( 1 \right)}}\ {\text{and }}\rho \left( x \right) = \frac{{\Psi '\left( x \right)}}{{\Psi '\left( 1 \right)}} = {e^{ - G\Lambda '\left( 1 \right)\left( {1 - x} \right)}}.
\end{align}

Each iteration of SIC can be divided into the decoding phase and the subtracting phase. Furthermore, let $p_i$ and $q_i$ be the average probabilities that a randomly selected edge can not be removed after the decoding phase and subtracting phase at the $i$-th iteration, respectively. For a degree-$r$ UN, an existing edge can be removed in the subtracting phase of the $i$-th iteration if any of the other $r-1$ edges has been removed after the decoding phase of the $(i-1)$-th iteration. Thus, we have
\begin{align}
\label{q_p}
	{q_i} = \sum\limits_r {{\lambda _r}p_{i - 1}^{r - 1}}  = \lambda \left( {{p_{i - 1}}} \right).
\end{align}
For a degree-$l$ SN, an existing edge can be removed in the decoding phase of the $i$-th iteration if there are no more than $K$ edges connected to this SN after the previous phase. Based on \cite{Mulit1}, we have
\begin{align}
\label{p_q}
	{p_i}
&=  \sum\limits_l {{\rho _l}} \left( {1 - \sum\limits_{k = 0}^{\text {min}(K,l)-1} {\dbinom{l-1}{k}}{q_i^k{{\left( {1 - {q_i}} \right)}^{l - k - 1}}} } \right) \notag\\
	&=  1 - \sum\limits_{k = 0}^{K-1} {\frac{{q_i^k{\rho ^{\left( k \right)}}\left( {1 - {q_i}} \right)}}{{k!}}}.
\end{align}
By substituting (\ref{edge_distribution}) and (\ref{q_p}) into (\ref{p_q}), it can be seen that
\begin{align}
\label{p_p}
	{p_i} = 1 - \exp \left( { - G\Lambda '\left( {{p_{i - 1}}} \right)} \right)\sum\limits_{k = 0}^{K - 1} {\frac{{{{\left( {G\Lambda '\left( {{p_{i - 1}}} \right)} \right)}^k}}}{{k!}}}.
\end{align}	

\section{Optimal Distribution Design}
\allowdisplaybreaks
It is clear that SIC stops when $p_i=p_{i-1}$. By substituting $p_i=p_{i-1}$ into (\ref{p_p}) and omitting the subscript of $p_{i-1}$, we
have
\begin{align}
\label{SICstop_K}
p = 1 - \exp \left( { - G\Lambda '\left( p \right)} \right)\sum\limits_{k = 0}^{K - 1} {\frac{{{{\left( {G\Lambda '\left( p \right)} \right)}^k}}}{{k!}}}.
\end{align}
If (\ref{SICstop_K}) has no real root for $0<p\le1$, the probability that a randomly selected packet cannot be successfully decoded after SIC converges to $0$, which is the basic condition that an optimal distribution is expected to meet.

\subsection{Optimal Distribution for the Case of $K=2$}
By substituting $K=2$ into (\ref{SICstop_K}), we obtain
\begin{align}
\label{SICstop}
G\Lambda '(p) - \ln \left( {1 + G\Lambda '(p)} \right) =  - \ln \left( {1 - p} \right).
\end{align}
Let us present the functional form of (\ref{SICstop}) as
\begin{align}
\label{f_p}
f(p) = G\Lambda '(p) - \ln \left( {1 + G\Lambda '(p)} \right) + \ln \left( {1 - p} \right).
\end{align}
Then, the condition to be satisfied is equivalent to that $f\left(p\right)$ has no zero point for $0<p\le1$. In other words, the derived optimal UN degree distribution should ensure that $f\left(p\right) < 0$ (or $f\left(p\right) > 0$) always holds for $0 < p \le 1$. The exponential function $g\left( p \right) = \exp \left( {ap} \right) - 1$ is selected as the approximation object of $G{\Lambda '\left( p \right)}$. Note that $a > 0$ is an adjustable positive parameter. By replacing $G{\Lambda '\left( p \right)}$ in (\ref{f_p}) with $g\left( p \right)$, we have
\begin{align}
	\label{tilde_f}
\tilde f\left(p\right) := {e^{ap}} - ap + \ln \left( {1 - p} \right) - 1,
\end{align}
where the assignment of $a$ is expected to ensure that $\tilde f\left( p \right) < 0$ always holds for $0<p\le1$. Since it can not be realized within the acceptable value range of $a$, we do not consider the case that $\tilde f\left( p \right) > 0$ always holds.

\begin{table}[!t]
	\centering
	\begin{tabular}{l}
		\toprule
		\bf{Algorithm 1:} Find optimal parameter ${a^*}$\\
		\midrule
		{\bf{Input:}} $0<p<1$, the value range of the unresolved probability for SIC\\ process; $a=0$, the initial value of the parameter; $\varepsilon=0.1$ and $\varepsilon^*$, the \\initial and objective values of parameter precision, respectively.  \\
				
		{\bf{Output:}} Optimal parameter ${a^*}$.\\
		{\bf{1:}} $a=a+\varepsilon$, i.e., gradually increase $a$ to approach $a^*$;\\
        {\bf{2:}} {\bf{if}} $\tilde{f}(p)$ decreases monotonically for $0<p<1$ {\bf{then}} \\
        ~~~~~return to step 1;\\
        ~~~{\bf{else}} \\
        ~~~~~go to step 3;\\
		{\bf{3:}} take the maximal local maximum of $\tilde{f}(p)$, which is denoted by $\tilde{f}_{\text {max}}$.\\
		~~~{\bf{if}} $\tilde{f}_{\text {max}}<0$ {\bf{then}} \\
		~~~~~return to step 1;\\
		~~~{\bf{else}} \\
		~~~~~go to step 4;\\
		{\bf{4:}} $a=a-\varepsilon$, i.e., back to the previous $a$ when $\tilde{f}_{\text {max}}>0$;\\
		{\bf{5:}} $\varepsilon=\varepsilon/10$, i.e., update $\varepsilon$ to determine the next digit of $a^*$;\\
        {\bf{6:}} {\bf{if}} $\varepsilon<\varepsilon^*$ {\bf{then}} \\
        ~~~~~return to step 1, i.e., continue to increase $a$ when $\varepsilon$ does not exceed \\the objective precision $\varepsilon^*$;\\
        ~~~{\bf{else}} \\
        ~~~~~go to step 7;\\
        {\bf{7:}} $a^*=a$, i.e., the $a$ at step 3 is exactly the optimal parameter $a^*$ if $\varepsilon$ \\has reached the objective precision $\varepsilon^*$.\\
        \bottomrule
    \end{tabular}
\end{table}

The derivative of (\ref{tilde_f}) with respect to $p$ is given by
\begin{align}
\label{f'p}
\tilde f'\left( p \right) = a{e^{ap}} - a - {\left( {1 - p} \right)^{ - 1}},
\end{align}
from which we have $\tilde f'\left( 0 \right) = -1$ and $\tilde f'\left( p \right) \to  - \infty$ if $p \to 1^-$. Besides, it is easy to know $\tilde f\left( 0 \right) = 0$ and $ \tilde f\left( p \right) \to  - \infty$ if $p \to 1^-$ from (\ref{tilde_f}). Therefore, if $\tilde f\left(p\right)$ has no extremum for $0<p<1$, then it is monotonically decreasing and obviously satisfies $\tilde f\left(p\right)<0$. In other cases, the local minimum and local maximum values of $\tilde f\left(p\right)$ for $0<p<1$ always appear in pairs, i.e., there is a local maximum after each local minimum along with the increase of $p$ from $0$ to $1$. To satisfy $\tilde f\left(p\right)<0$, all local maximums of $\tilde f\left(p\right)$ are required to be less than $0$.
The derivative of (\ref{tilde_f}) with respect to $a$ is given by
\begin{align}
\label{f'a}
	\tilde f'\left( a \right) = p\left( {{e^{ap}} - 1} \right),
\end{align}
which is obviously greater than $0$ for $a>0$ and $0<p\le 1$. Thus, each function value of $\tilde f\left(p\right)$ increases as $a$ increases. That is to say, if we find an ${a^*}$ that the corresponding maximum local maximum of $\tilde f\left(p\right)$ in $(0,1)$ infinitely approaches $0$, then for any $a>a^*$, it can not make $\tilde f\left(p\right) < 0$ always hold for $0<p\le1$.
Due to the transcendence of (\ref{tilde_f}), it is difficult to analytically obtain ${a^*}$. As an alternative method, we use Algorithm 1, the principle of which is to determine each digit of ${a^*}$ successively from high to low by continuously increasing $a$. Setting $\varepsilon^*=0.01$ in Algorithm 1, $a^*=1.73$ is obtained. Fig. 2 shows the comparison between $\tilde f\left(p\right)$ and $f\left(p\right)$.

\begin{thm}
For any positive integer $L$, consider the UN degree distribution for IRSA with K = 2 as
\begin{align}
	\label{Lambda1}
	\Lambda_1 (x) = \sum\limits_{s = 2}^{L + 1} {\frac{{{{1.73}^{s - 1}}/\left( {s!} \right)}}{{\sum\nolimits_{t = 1}^L {{{1.73}^t}/\left( {\left( {t + 1} \right)!} \right)} }}{x^s}}.
\end{align}
Then, the probability that one packet cannot be successfully decoded after SIC converges to $0$ if the load satisfies that
\begin{align}
\label{G_threshold}
G  \le \sum\limits_{t = 1}^L {\frac{{{1.73^t}}}{{(t + 1)!}}} .
\end{align}
\end{thm}
\begin{proof}
Let $\tilde \Lambda ^\prime_1 \left( p \right)$ be the $L$th-order Taylor expansion of $g\left( p \right) = \exp \left( {ap} \right) - 1$ at $p=0$ without the Lagrange remainder ${{\left( {{{\left( {ap} \right)}^{L + 1}}\exp \left( {a\xi } \right)} \right)} \mathord{\left/{\vphantom {{\left( {{{\left( {ap} \right)}^{L + 1}}\exp \left( {a\xi } \right)} \right)} {\left( {\left( {L + 1} \right)!} \right)}}} \right.\kern-\nulldelimiterspace} {\left( {\left( {L + 1} \right)!} \right)}}$, where $\xi \in \left( {0,p} \right)$ is a constant. For $a>0$ and $0<p\le1$, it is obviously that the Lagrange remainder is positive. Thus, we have $\tilde \Lambda ^\prime_1 \left( p \right) = \sum\nolimits_{s = 1}^L {{{{{\left( {ap} \right)}^s}} \mathord{\left/{\vphantom {{{{\left( {ap} \right)}^s}} {\left( {s!} \right)}}} \right.\kern-\nulldelimiterspace} {\left( {s!} \right)}}}  < g\left( p \right)$. In addition, It is easy to know that $h\left( x \right) = x - \ln \left( {1 + x} \right)$ increases monotonically for $x>0$. Therefore, $h\left( {g\left( p \right)} \right) > h\left( {{{\tilde \Lambda^\prime_1}}\left( p \right)} \right)$. On the other hand, by substituting (\ref{Lambda1}) and (\ref{G_threshold}) into (\ref{f_p}), we obtain
\begin{align}
\label{f_p<}
f\left(p\right) &= G{{\Lambda {'_1}\left( p \right)}} - \ln \left( 1 + G{{\Lambda {'_1}\left( p \right)}} \right) + \ln \left({1 - p}\right) \notag \\
& \le {\tilde \Lambda^\prime_1}\left( p \right) - \ln \left( {1 + {{\tilde \Lambda^\prime}_1}\left( p \right)} \right) + \ln \left({1 - p}\right),
\end{align}
where the inequality follows from setting $a = a^* = 1.73$ and $h\left( G{\Lambda {'_1}\left( p \right)} \right) \le h\left( {{{\tilde \Lambda^\prime}_1}\left( p \right)} \right)$. Based on (\ref{tilde_f}) and (\ref{f_p<}), we have
\begin{align}
\tilde f\left( p \right) - f\left( p \right) \ge  h\left( {g\left( p \right)} \right) - h\left( {{{\tilde \Lambda^\prime}_1}\left( p \right)} \right)>0,
\end{align}
that is, $\tilde f\left(p\right) > f\left( p \right)$. Since $a = a^*=1.73$ guarantees $\tilde f\left(p\right)< 0$, we can further obtain $f\left( p \right)< 0$. The proof completes.
\end{proof}
\begin{figure}
\begin{center}
\includegraphics[width=0.393\textwidth]{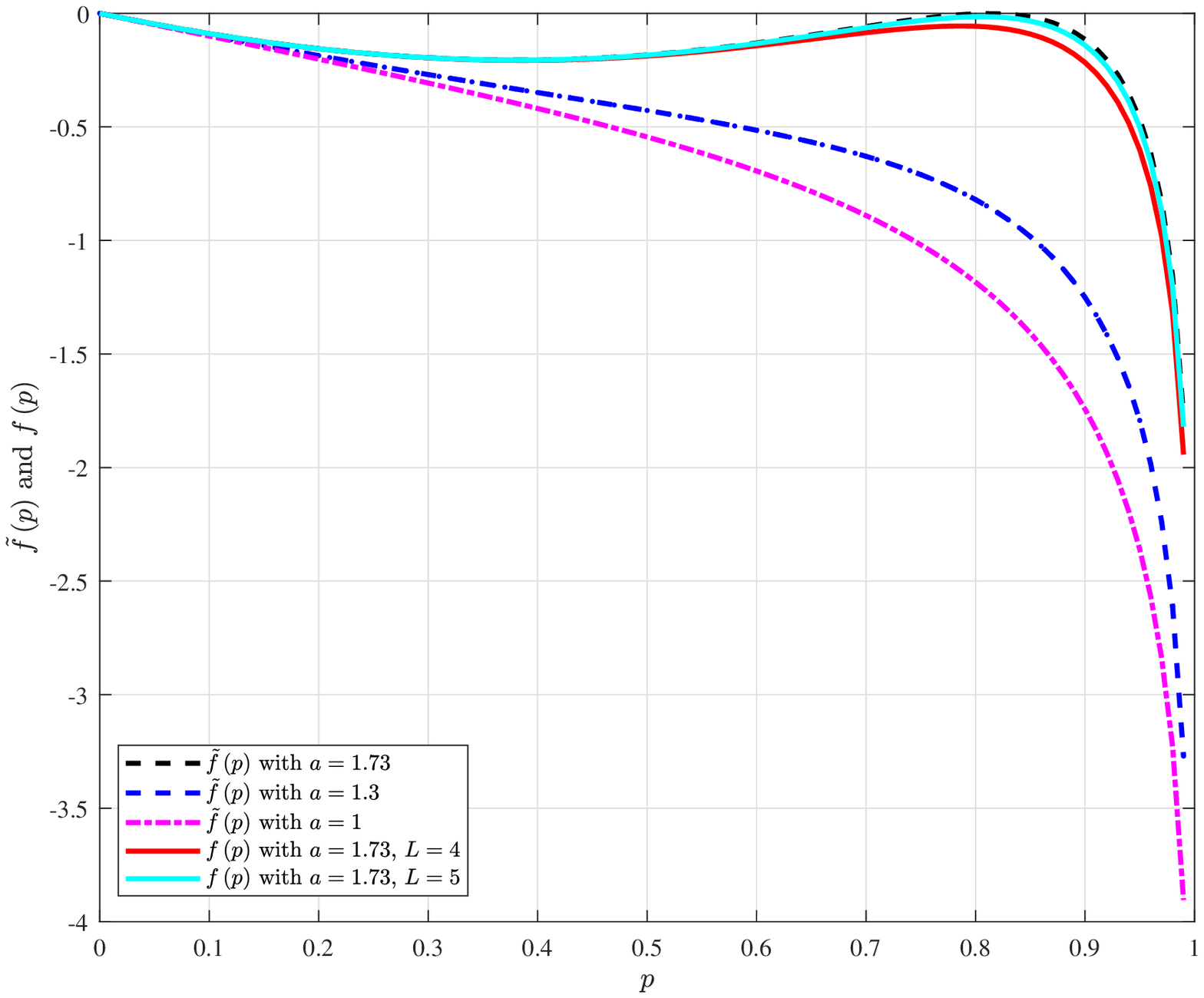}
\end{center}
\caption{$\tilde f\left(p\right)$ and $f\left(p\right)$.}
\end{figure}
Theorem 1 illustrates that $\Lambda_1 (x)$ enables the system load to achieve $G^* =1.68$. Due to the fast convergence of Taylor expansion $\tilde \Lambda ^\prime_1 \left( p \right)$ to $g\left(p\right)$, this load threshold can be reached when $L= 5$, thereby allowing a small average number of packet replicas ($2.74$). In addition, it should be pointed out that the increase in $G^*$ brought about by setting the precision $\varepsilon^*$ in Theorem 1 to be finer (e.g., $\varepsilon^*=0.001$) is negligible.

\subsection{Discussions for the Cases of $K\ge3$}
\allowdisplaybreaks
Two possible methods for deriving the UN degree distributions for $K \ge 3$ are presented below. The first is similar to the method of obtaining $\Lambda_1 (x)$. Taking the case of $K = 3$ as an example, we substitute $K = 3$ into (\ref{SICstop_K}) and have
\begin{align}
G\Lambda '(p) - \ln \left( {1 + G\Lambda '(p) + \frac{1}{2}{{\left( {G\Lambda '(p)} \right)}^2}} \right) = \ln \frac{1}{{1 - p}},
\end{align}
which has the same form as (\ref{SICstop}). Similar to the case of $K = 2$, let
\begin{align}
\label{K3}
1 + \tilde \Lambda '\left( p \right) + \frac{1}{2}{\left( {\tilde \Lambda '\left( p \right)} \right)^2} \approx \exp \left( {ap} \right),
\end{align}
where $\tilde \Lambda '\left( p \right) = G\Lambda '(p)$ and $a$ is a parameter. By expressing $\exp\left(ap\right)$ in Taylor expansion and taking the approximation in (\ref{K3}) as equality, we get
\begin{align}
{\left( {\tilde \Lambda '\left( p \right) + 1} \right)^2} = {p^0} + 2\sum\limits_{s = 1}^L  {\frac{{{a^s}}}{{s!}}} {p^s},
\end{align}
where $L$ is the order of the expansion. On the one hand, polynomial multiplication can be used to express linear convolution, thus we can think of ${p^0} + 2\sum\nolimits_{s = 1}^l {{{{{\left( {ap} \right)}^s}} \mathord{\left/
{\vphantom {{{{\left( {ap} \right)}^s}} {\left( {s!} \right)}}} \right.\kern-\nulldelimiterspace} {\left( {s!} \right)}}} $ as the convolution of $\tilde \Lambda '\left( p \right) + 1$ and itself. On the other hand, the circular convolution of two time domain sequences and the product of their corresponding frequency domain sequences are discrete Fourier transform pairs. Therefore, inversely transforming the square root of the discrete Fourier transform of ${p^0} + 2\sum\nolimits_{s = 1}^l {{{{{\left( {ap} \right)}^s}} \mathord{\left/{\vphantom {{{{\left( {ap} \right)}^s}} {\left( {s!} \right)}}} \right.\kern-\nulldelimiterspace} {\left( {s!} \right)}}} $ as the convolution of ˜$\tilde \Lambda '\left( p \right) + 1$ seems to be a feasible way to obtain $\tilde \Lambda '\left( p \right)$. Note that the first thing to be solved is the conversion problem between linear convolution and circular convolution.

The above method related to convolution is no longer applicable for $K \ge 4$. In fact, the larger the value of $K$, the more difficult it is to obtain $\Lambda (x)$. Another more universal method is to use optimization algorithms similar to Algorithm 1. For any $K$, we directly replace $G{\Lambda '\left( p \right)} $ with $\exp \left( {ap} \right) - 1$ and give $\tilde f(p)$ similar to (\ref{tilde_f}) as
\begin{align}
\tilde f(p) = {e^{ap}} - \ln \left( {\sum\limits_{k = 0}^{K - 1} {\frac{{{{\left( {{e^{ap}} - 1} \right)}^k}}}{{k!}}} } \right) + \ln \left( {{1 - p}} \right) - 1.
\end{align}
After studying the existence, $a^*$ that makes $\tilde f(p)$ the closest to the horizontal axis but dose not intersect with it can be found by algorithms. In fact, it is not necessary to use an exponential function such as $\exp \left( {ap} \right) - 1$. Although the fast convergence of the Taylor expansion is an advantage of exponential functions, the fixed function form may limit the performance of the UN degree distributions obtained.

\section{Energy Efficiency Optimization}
\allowdisplaybreaks
The energy efficiency optimization in terms of the maximum repetition rate $L$ (actually $L + 1)$ for $\Lambda_1\left(x\right)$ is presented in this section.
Let ${P_{\text {tx}}}$ and ${P_{\text c}}$ be the powers consumed by each user to transmit data packets or not in a time slot, respectively. In addition, the duration of each slot is assumed to be $1$. Thus, the average energy consumption is $E = \Lambda '(1){P_{\text {tx}}} + N{P_{\text c}}$, where $\Lambda '(1)$ denotes the average number of replicas transmitted by a user during one frame. Shannon capacity is used to approximate the information transmission rate. Note that, due to the repetition strategy of IRSA scheme, the number of valid packets belonging to each user in a frame is actually one. Correspondingly, the effective amount of data transmitted by each user in a frame per unit bandwidth is ${\log _2}\left( {1 + {{{P_{\text {tx}}}} \mathord{\left/{\vphantom {{{P_{\text {tx}}}} {{\sigma ^2}}}} \right.\kern-\nulldelimiterspace} {{\sigma ^2}}}} \right)$, where $\sigma ^2$ denotes the average noise power. Therefore, the energy efficiency is given by
\begin{align}
	\label{EE}
	\Gamma  = \frac{{{{\log }_2}\left( {1 + {{{P_\text {tx}}} \mathord{\left/
						{\vphantom {{{P_\text {tx}}} {{\sigma ^2}}}} \right.
						\kern-\nulldelimiterspace} {{\sigma ^2}}}} \right)}}{E} = \frac{{{{\log }_2}\left( {1 + {{{P_\text {tx}}} \mathord{\left/
						{\vphantom {{{P_\text {tx}}} {{\sigma ^2}}}} \right.
						\kern-\nulldelimiterspace} {{\sigma ^2}}}} \right)}}{{\Lambda '(1){P_\text {tx}} + N{P_\text c}}}.
\end{align}

\begin{table}
\centering
\setlength\tabcolsep{5.9pt}
\caption{The first $7$ $\Delta {A_L}/\left| {\Delta {B_L}} \right|$ values of the proposed distribution}
\begin{tabular}{cccccccc}
\toprule
$L$&$1$&$2$&$3$&$4$&$5$&$6$&$7$\\
\midrule
$\frac{{\Delta {A_L}}}{{\left| {\Delta {B_L}} \right|}}$&$0.8649$&$2.2298$&$3.8042$&$5.5065$&$7.0526$&$7.2$&$9$\\
\bottomrule
\end{tabular}
\end{table}

Assume that $M$ (the number of active users) is given and $G^*$ (the optimal load) is adopted. Based on (\ref{G}) and (\ref{G_threshold}), the number of slots is expressed as $N = {M \mathord{\left/{\vphantom {M {\sum\nolimits_{t = 1}^L {{{1.73}^t}/\left( {\left( {t + 1} \right)!} \right)} }}} \right.\kern-\nulldelimiterspace} {\sum\nolimits_{t = 1}^L {{{1.73}^t}/\left( {\left( {t + 1} \right)!} \right)} }}$. Furthermore, $\Lambda '(1)$ can be obtained via (\ref{Lambda1}). Thus, we obtain
\begin{align}
\label{E}
 E = {P_{\text {tx}}}\frac{{\sum\nolimits_{s = 2}^{L + 1} {{{1.73}^{s - 1}}/\left( {\left( {s - 1} \right)!} \right)}  + M{P_\text c}/{P_{\text {tx}}}}}{{\sum\nolimits_{t = 1}^L {{{1.73}^t}/\left( {\left( {t + 1} \right)!} \right)} }}.
\end{align}
In order to discuss the effect of $L$ on $E$, we write (\ref{E}) as ${E_L} = {P_{\text {tx}}}\left( {{A_L} + {B_L}\left( {{{M{P_{\text c}}} \mathord{\left/
{\vphantom {{M{P_{\text c}}} {{P_{\text {tx}}}}}} \right.
\kern-\nulldelimiterspace} {{P_{\text {tx}}}}}} \right)} \right)$. It can be found that $B_L$ decreases monotonically as $L$ increases. As for $A_L$, we have
\begin{align}
\label{A}
	{A_{L + 1}} - {A_L}  = \frac{{\frac{{{{1.73}^{L + 1}}}}{{\left( {L + 1} \right)!}}\left( {\sum\nolimits_{t = 1}^L {\frac{{{{1.73}^t}}}{{t!\left( {t + 1} \right)}}}  - \sum\nolimits_{s = 2}^{L + 1} {\frac{{{{1.73}^{s - 1}}}}{{\left( {s - 1} \right)!\left( {L + 2} \right)}}} } \right)}}{{\left( {\sum\nolimits_{t = 1}^{L + 1} {\frac{{{{1.73}^t}}}{{\left( {t + 1} \right)!}}} } \right)\left( {\sum\nolimits_{t = 1}^L {\frac{{{{1.73}^t}}}{{\left( {t + 1} \right)!}}} } \right)}},
\end{align}
which is greater than $0$. The reason is that ${t + 1}< {L + 2}$ in the numerator holds true for each items added. Thus $A_L$ increases monotonically as $L$ increases, which is the opposite of $B_L$. Therefore, the magnitude relationship between the changes of $A_L$ and $B_L$, denoted as $\Delta {A_L} = {A_{L + 1}} - {A_L}$ and $\Delta {B_L}={B_{L + 1}} - {B_L}$, affects the trend of $E$. Specifically, if ${{{\Delta {A_L}} \mathord{\left/
{\vphantom {{\Delta {A_L}} {\left| {\Delta {B_L}} \right|}}} \right.\kern-\nulldelimiterspace} {\left| {\Delta {B_L}} \right|}}}< {{{M{P_\text c}} \mathord{\left/{\vphantom {{M{P_\text c}} {{P_{\text {tx}}}}}} \right.\kern-\nulldelimiterspace} {{P_{\text {tx}}}}}}$, then ${E_{L+1}}<{E_L}$. Similarly, if  ${{{\Delta {A_L}} \mathord{\left/{\vphantom {{\Delta {A_L}} {\left| {\Delta {B_L}} \right|}}} \right.
\kern-\nulldelimiterspace} {\left| {\Delta {B_L}} \right|}}}>  {{{M{P_\text c}} \mathord{\left/{\vphantom {{M{P_\text c}} {{P_{\text {tx}}}}}} \right.\kern-\nulldelimiterspace} {{P_{\text {tx}}}}}} $, then ${E_{L+1}}>{E_L}$. Furthermore, the general expression of ${{{\Delta {A_L}} \mathord{\left/{\vphantom {{\Delta {A_L}} {\left| {\Delta {B_L}} \right|}}} \right.\kern-\nulldelimiterspace} {\left| {\Delta {B_L}} \right|}}}$ is given by
\begin{align}
\label{A/B}
\frac{{\Delta {A_L}}}{{\left| {\Delta {B_L}} \right|}} = \sum\limits_{i = 1}^L {\frac{{{{1.73}^i}\left( {L + 1 - i} \right)}}{{\left( {i + 1} \right)!}}}.
\end{align}
Since each term accumulated in (\ref{A/B}) is positive, ${{{\Delta {A_L}} \mathord{\left/{\vphantom {{\Delta {A_L}} {\left| {\Delta {B_L}} \right|}}} \right.\kern-\nulldelimiterspace} {\left| {\Delta {B_L}} \right|}}}$ increases monotonically with respect to $L$. Thus, for
\begin{align}
\label{A^*/B^*}
	\frac{{\Delta {A_{{L^*} - 1}}}}{{\left| {\Delta {B_{{L^*} - 1}}} \right|}} < \frac{{M{P_{\text c}}}}{{{P_{\text {tx}}}}} < \frac{{\Delta {A_{{L^*}}}}}{{\left| {\Delta {B_{{L^*}}}} \right|}},
\end{align}
we can know that $E_L$ decreases when $L<L^*$ and increases when $L>L^*$. That is to say, $L^*$ maximizes the energy efficiency  $\Gamma$ (i.e., minimizes the energy consumption $E$).

Substituting $L=1$ into (\ref{A/B}), we have ${{{\Delta {A_1}} \mathord{\left/{\vphantom {{\Delta {A_1}} {\left| {\Delta {B_1}} \right|}}} \right.\kern-\nulldelimiterspace} {\left| {\Delta {B_1}} \right|}}} = 0.8649$. Therefore, as long as ${{{M{P_\text c}} \mathord{\left/{\vphantom {{M{P_\text c}} {{P_{\text {tx}}}}}} \right.\kern-\nulldelimiterspace} {{P_{\text {tx}}}}}} > 0.8649$, the energy efficiency of the system always has a peak value, which can be obtained by setting $L = L^*$. We list the first few values of ${{\Delta {A_L}} \mathord{\left/{\vphantom {{\Delta {A_L}} {\left| {\Delta {B_L}} \right|}}} \right.\kern-\nulldelimiterspace} {\left| {\Delta {B_L}} \right|}}$ in Tab. \uppercase\expandafter{\romannumeral1} for reference.

\section{Numerical Results}
\allowdisplaybreaks
In this section, numerical results are presented to discuss the throughput and energy efficiency performance of the proposed optimal transmission probability distribution (i.e., ${\Lambda _1}\left( x \right)$ with $a=1.73$). Specifically, the results of the packet loss rate (PLR, defined as the probability that any packet replicas of a user cannot be successfully decoded) versus load and the energy efficiency versus maximum repetition rate are provided.
\begin{figure}
\begin{center}
\includegraphics[width=0.393\textwidth]{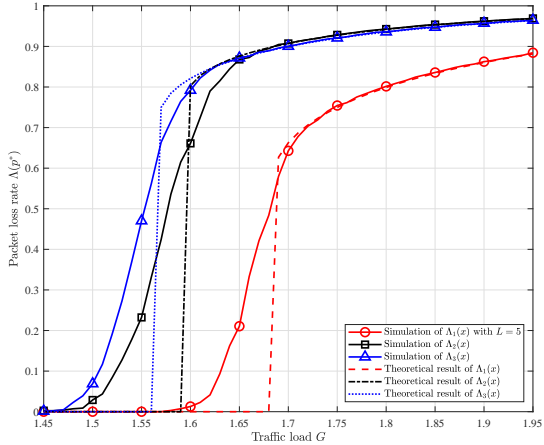}
\end{center}
\caption{Packet loss rate versus traffic load for IRSA with $K=2$. ${\Lambda _2}\left( x \right) = 0.5{x^2} + 0.28{x^3} + 0.22{x^8}$ and ${\Lambda _3}\left( x \right) = 0.25{x^2} + 0.60{x^3} + 0.15{x^8}$.}
\end{figure}

Fig. 3 shows the PLR versus traffic load for IRSA with $K=2$. Let $p^*$ be the largest root of (\ref{SICstop}), then PLR can be obtained by $\Lambda\left(p^*\right)$. The number of users is fixed to $M=1000$ and the load $G$ is changed by adjusting the number of time slots in a frame. In addition to the derived ${\Lambda _1}\left( x \right)$, the two distributions ${\Lambda _2}\left( x \right)$ and ${\Lambda _3}\left( x \right)$ introduced in \cite{IRSA} are also presented as references. It can be seen that for the case of $K=2$, ${\Lambda _1}\left( x \right)$ can achieve a significantly higher load threshold and better robustness. Since the theoretical results are obtained by solving (\ref{SICstop}) based on the asymptotic model (where $M,N \to \infty$), there is a gap with the simulation results.

\begin{figure}
\begin{center}
\includegraphics[width=0.393\textwidth]{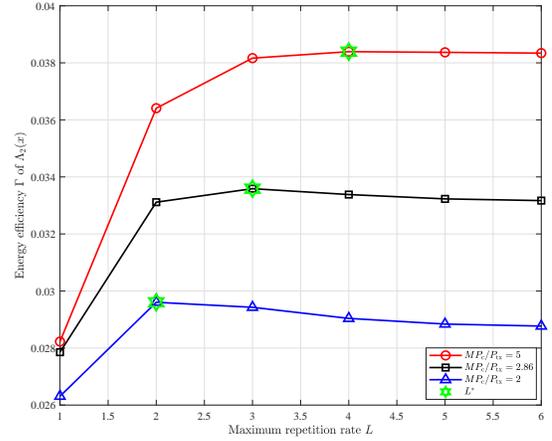}
\end{center}
\caption{Energy efficiency versus maximum repetition rate for ${\Lambda _1}\left( x \right)$.}
\end{figure}

Further setting ${\sigma ^2} = 1$ and $P_{\text c}=0.1$, Fig. 4 shows the energy efficiencies versus maximum repetition rate for ${\Lambda _1}\left( x \right)$. We select the three $P_{\text {tx}}$ values of $20$, $35$, and $50$ to get different ${{M{P_\text c}} \mathord{\left/{\vphantom {{M{P_\text c}} {{P_{\text {tx}}}}}} \right.\kern-\nulldelimiterspace} {{P_{\text {tx}}}}}$. For ${M{P_\text c}/{P_{\text {tx}}}} = 2$, $2.86$, and $5$, the corresponding values of $L^*$ are $2$, $3$, and $4$, respectively. The same results can also be obtained by directly referring to Tab \uppercase\expandafter{\romannumeral1}. Furthermore, it can be seen from Fig. 4 that although it still exists, the peak value of the curve becomes less obvious as the transmit power decreases. This is because the decrease of $P_{\text {tx}}$ means the increase of $L^*$, and the fast convergence of (\ref{E}) makes the energy consumption change very slowly after $L>4$.

\section{Conclusion}
\allowdisplaybreaks

An optimal transmission probability distribution was analytically derived for IRSA with a multi-packet reception capability of $2$, as $K=2$ is promising in achieving the highest time-frequency  resource efficiency in multi-packet reception scenarios. In addition to the traceable expression, the derived transmission distribution also exhibits a better performance than the existing benchmark distributions in terms of the throughput. To achieve a high energy efficiency, the optimal repetition rate is also found based to the derived transmission distribution. It is shown that using the optimized repetition rate parameter, the energy efficiency can be improved significantly.

\ifCLASSOPTIONcaptionsoff
  \newpage
\fi



%
\bibliographystyle{IEEEtran}
\bibliography{K=2}

\end{document}